\documentclass[letterpaper, 10 pt, conference]{ieeeconf}  
\IEEEoverridecommandlockouts                              
\usepackage[ruled,vlined]{algorithm2e}
\usepackage{bm}
\usepackage{amsmath}

\usepackage{amsthm}

\usepackage[utf8]{inputenc} 
\usepackage[T1]{fontenc}    
\usepackage{hyperref}       
\usepackage{url}            
\usepackage{booktabs}       
\usepackage{amsfonts}       
\usepackage{nicefrac}       
\usepackage{microtype}      
\usepackage{cleveref}       
\usepackage{lipsum}         
\usepackage{graphicx}
\usepackage{multirow}       
\usepackage{doi}
\usepackage{threeparttable}
\usepackage{xcolor}
\usepackage{subcaption}

\newtheorem{theorem}{Theorem}

\crefname{lemma}{lemma}{lemmas}
\Crefname{lemma}{Lemma}{Lemmas}
\crefname{assumption}{assumption}{assumptions}
\Crefname{assumption}{Assumption}{Assumptions}
\crefname{definition}{definition}{definitions}
\Crefname{definition}{Definition}{Definitions}

\makeatletter
\def\input@path{{sections/}}
\makeatother

\graphicspath{{images/}}

\title{\LARGE \bf Scalable Co-Clustering for Large-Scale Data through Dynamic Partitioning and Hierarchical Merging}

\author{Zihan Wu$^{1}$, Zhaoke Huang$^{2}$, and Hong Yan$^{3}$
\thanks{This work is supported by Hong Kong Innovation and
Technology Commission (InnoHK Project CIMDA) and Hong
Kong Research Grants Council (Project CityU 11204821).}
\thanks{$^{1}$Zihan Wu (Corresponding Author) is with the Department of Electrical Engineering,
City University of Hong Kong, Hong Kong
{\tt\small zihan.wu@my.cityu.edu.hk}}%
\thanks{$^{2}$Zhaoke Huang is with the Department of Electrical Engineering,
City University of Hong Kong, Hong Kong
{\tt\small Z.Huang@cityu.edu.hk}}%
\thanks{$^{3}$Hong Yan is with the Department of Electrical Engineering,
City University of Hong Kong, Hong Kong
{\tt\small h.yan@cityu.edu.hk}}%
\thanks{This is the author's accepted version of the paper. The final published version appears in the proceedings of the \textit{2024 IEEE International Conference on Systems, Man, and Cybernetics (SMC 2024)}. 
Published in: \textit{2024 IEEE International Conference on Systems, Man, and Cybernetics (SMC)}
Conference Date: 6--10 October 2024
Added to IEEE Xplore: 20 January 2025 
DOI: \href{https://doi.org/10.1109/SMC54092.2024.10832071}{10.1109/SMC54092.2024.10832071}
\textcopyright\ 2024 IEEE. Personal use of this material is permitted. Permission from IEEE must be obtained for all other uses, including reprinting/republishing this material for advertising or promotional purposes, creating new collective works for resale or redistribution to servers or lists, or reuse of any copyrighted components of this work in other works.}
}

\hypersetup{
        pdftitle={Scalable Co-Clustering for Large-Scale Data through Dynamic Partitioning and Hierarchical Merging},
        pdfsubject={Artificial Intelligence, Machine Learning, Data Mining},
        pdfauthor={Zihan Wu, Zhaoke Huang, Hong Yan},
        pdfkeywords={Co-clustering, Large-scale data, Matrix partitioning, Hierarchical merging},
        }
        
\begin{document}
\maketitle

\thispagestyle{empty}
\pagestyle{empty}


\begin{abstract}
Co-clustering simultaneously clusters rows and columns, revealing more fine-grained groups. However, existing co-clustering methods suffer from poor scalability and cannot handle large-scale data. This paper presents a novel and scalable co-clustering method designed to uncover intricate patterns in high-dimensional, large-scale datasets. Specifically, we first propose a large matrix partitioning algorithm that partitions a large matrix into smaller submatrices, enabling parallel co-clustering. This method employs a probabilistic model to optimize the configuration of submatrices, balancing the computational efficiency and depth of analysis.
Additionally, we propose a hierarchical co-cluster merging algorithm that efficiently identifies and merges co-clusters from these submatrices, enhancing the robustness and reliability of the process. Extensive evaluations validate the effectiveness and efficiency of our method. Experimental results demonstrate a significant reduction in computation time, with an approximate 83\% decrease for dense matrices and up to 30\% for sparse matrices.

\end{abstract}

\section{Introduction}
Artificial Intelligence is a rapidly advancing technology facilitating complex data analysis, pattern recognition, and decision-making processes. Clustering, a fundamental unsupervised learning technique, groups data points based on shared features, aiding in interpreting complex data structures. However, traditional clustering algorithms \cite{zhang2023AdaptiveGraphConvolution, wu2023EffectiveClusteringStructured} treat all features of data uniformly and solely cluster either rows (samples) or columns (features),  as shown in Figure \ref{fig:cluster}. They oversimplified interpretations and overlooked critical context-specific relationships within the data, especially when dealing with large, high-dimensional datasets \cite{chen2023FastFlexibleBipartite, zhao2023MultiviewCoclusteringMultisimilarity, kumar2023CoclusteringBasedMethods}.

\textit{Co-clustering} \cite{kluger2003SpectralBiclusteringMicroarray, yan2017CoclusteringMultidimensionalBig} is a technique that groups rows (samples) and columns (features) simultaneously, as shown in Figure \ref{fig:cocluster}. It can reveal complex correlations between two different data types and is transformative in scenarios where the relationships between rows and columns are as important as the individual entities themselves. For example, in bioinformatics, co-clustering could identify gene-related patterns leading to biological insights by concurrently analyzing genes and conditions \cite{higham2007SpectralClusteringIts, kluger2003SpectralBiclusteringMicroarray, zhao2012BiclusteringAnalysisPattern}. In recommendation systems, co-clustering can simultaneously discover more fine-grained relationships between users and projects \cite{dhillon2007WeightedGraphCuts, chen2023ParallelNonNegativeMatrix}. Co-clustering extends traditional clustering methods, enhancing accuracy in pattern detection and broadening the scope of analyses.

\begin{figure}[htbp]
    \centering
    \begin{subfigure}[b]{0.22\textwidth}
        \includegraphics[width=\linewidth]{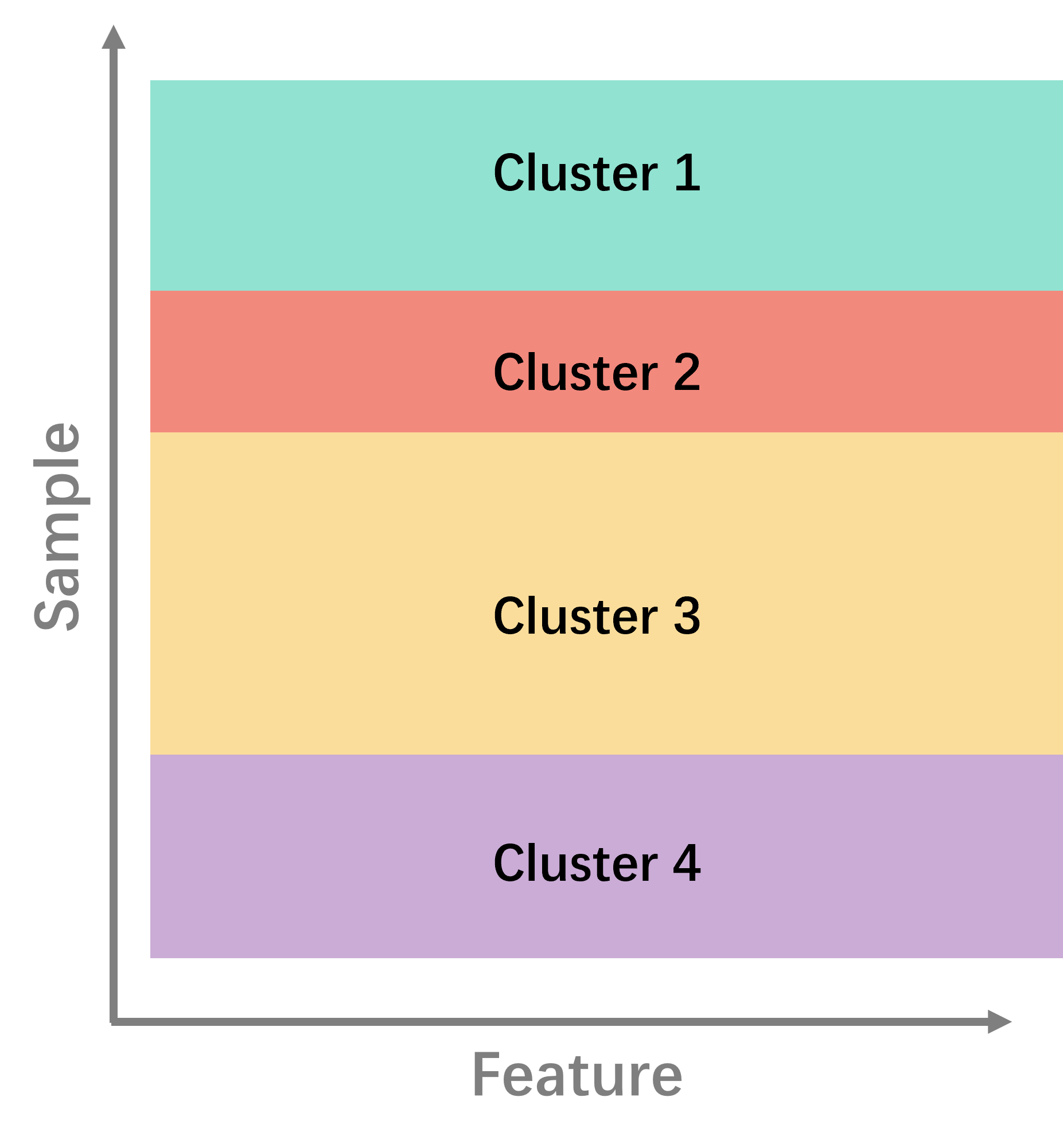}
        \caption{Clustering}
        \label{fig:cluster}
    \end{subfigure}
    \hfill
    \begin{subfigure}[b]{0.22\textwidth}
        \includegraphics[width=\linewidth]{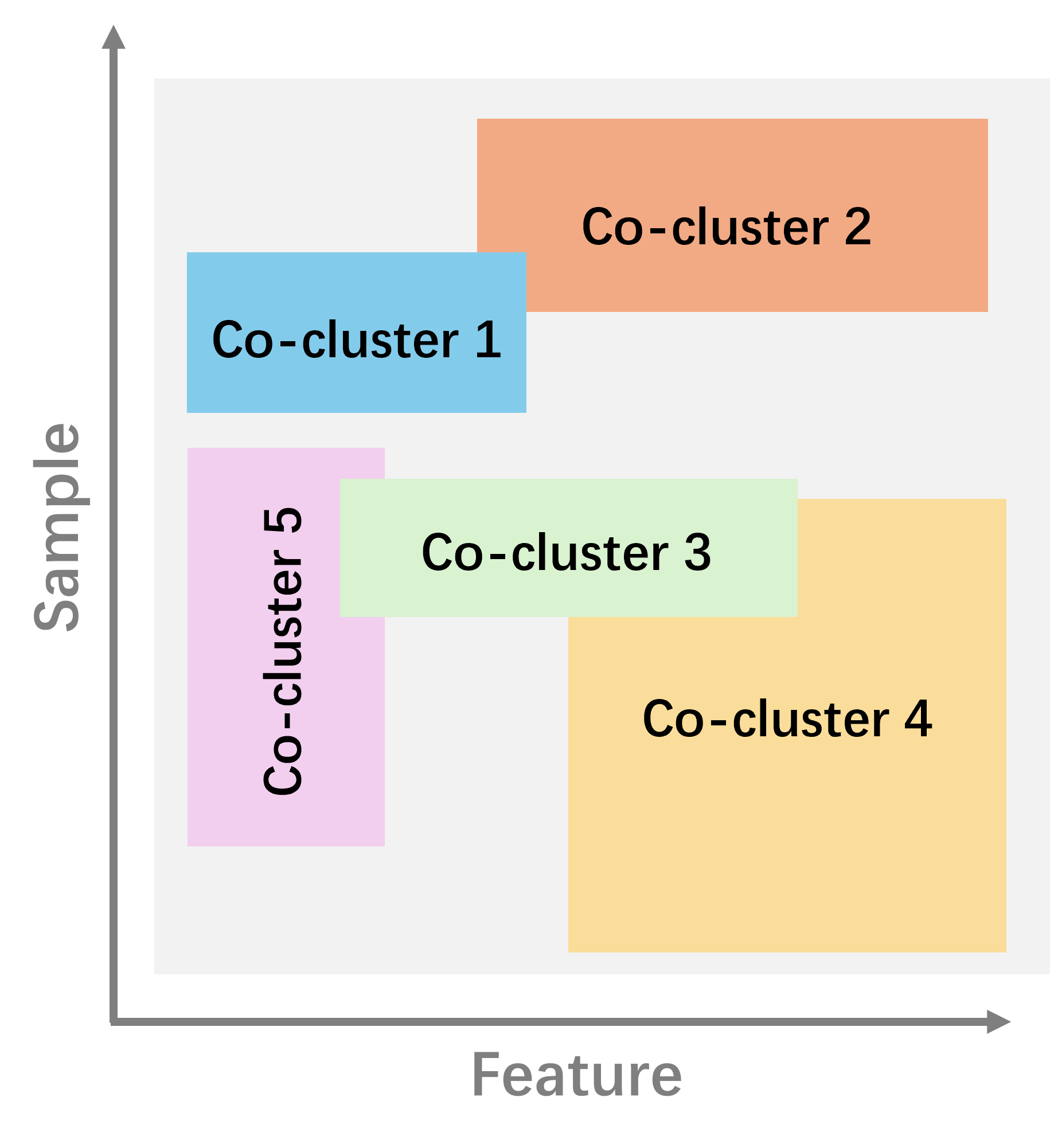}
        \caption{Co-clustering}
        \label{fig:cocluster}
    \end{subfigure}
    \caption{An illustration of the differences between (a) Clustering and (b) Co-clustering \cite{yan2017CoclusteringMultidimensionalBig}.}
    \label{fig:cocomparison}
\end{figure}

Despite its potential, scaling co-clustering to large datasets poses significant challenges:

\begin{itemize}
    \item{\textbf{High Computational Complexity.}} Co-clustering analyzes relationships both within and across the rows and columns of a dataset simultaneously. This dual-focus analysis requires evaluating a vast number of potential relationships, particularly as the dimensions of the data increase. The complexity can grow exponentially with the size of the data because the algorithm must process every possible combination of rows and columns to identify meaningful clusters \cite{hansen2011NonparametricCoclusteringLarge}.
    \item{\textbf{Significant Communication Overhead.}} Even when methods such as data partitioning are used to handle large-scale data, each partition may independently analyze a subset of the data. However, to optimize the clustering results globally, these partitions need to exchange intermediate results frequently. This requirement is inherent to iterative optimization techniques used in co-clustering, where each iteration aims to refine the clusters based on new data insights, necessitating continuous updates across the network. Such extensive communication can become a bottleneck, significantly slowing down the overall processing speed.
    \item{\textbf{Dependency on Sparse Matrices.}} Several traditional co-clustering algorithms are designed to perform best with sparse matrices \cite{pan2008CRDFastCoclusteringa}. However, in many real-world applications, data matrices are often dense, meaning most elements are non-zero. Such scenarios present a significant challenge for standard co-clustering algorithms, as they must handle a larger volume of data without the computational shortcuts available with sparse matrices.
\end{itemize}

To address the inherent challenges associated with existing co-clustering methods, we propose a novel and scalable Large-scale Adaptive Matrix Co-clustering (\textbf{LAMC}) framework designed for large-scale datasets. First,  we propose a large matrix partitioning algorithm that divides the original data matrix into smaller submatrices. This partitioning facilitates parallel processing of co-clustering tasks across submatrices, significantly reducing both processing time and computational and storage demands for each processing unit. We also design a probabilistic model to determine the optimal number and configuration of these submatrices to ensure comprehensive data coverage.
Second, we develop a hierarchical co-cluster merging algorithm that iteratively combines the co-clusters from these submatrices. This process enhances the accuracy and reliability of the final co-clustering results and ensures robust and consistent clustering performance, particularly addressing issues of heterogeneity and model uncertainty.

The contributions of this paper are summarized as follows:
\begin{enumerate}
    \item \textbf{Large Matrix Partitioning Algorithm:}
          We propose a novel matrix partitioning algorithm that enables parallel co-clustering by dividing a large matrix into optimally configured submatrices. This design is supported by a probabilistic model that calculates the optimal number and order of submatrices, balancing computational efficiency with the detection of relevant co-clusters.
    \item \textbf{Hierarchical Co-cluster Merging Algorithm:}
          We design a hierarchical co-cluster merging algorithm that combines co-clusters from submatrices, ensuring the completion of the co-clustering process within a pre-fixed number of iterations. This algorithm significantly enhances the robustness and reliability of the co-clustering process, effectively addressing model uncertainty.
    \item \textbf{Experimental Valuation:}
          We evaluate the effectiveness and efficiency of our method across a wide range of scenarios with large, complex data. Experimental results show an approximate 83\% decrease for dense matrices and up to 30\% for sparse matrices.
\end{enumerate}

The rest of this paper is organized as follows: Section \ref{sec:related_work} reviews related works; Section \ref{sec:formula} presents the problem formulation; Section \ref{sec:method} describes our LAMC method; Section \ref{sec:experiment} reports experimental results; and Section \ref{sec:conclude} concludes the paper.


\section{Related work}
\label{sec:related_work}
\subsection{Co-clustering Methods}
Co-clustering methods, broadly categorized into graph-based and matrix factorization-based approaches, have limitations in handling large datasets. Graph-based methods like Flexible Bipartite Graph Co-clustering (FBGPC) \cite{chen2023FastFlexibleBipartite} directly apply flexible bipartite graph models. Matrix factorization-based methods, such as Non-negative Matrix Tri-Factorization (NMTF) \cite{long2005CoclusteringBlockValue}, decompose data to cluster samples and features separately. Deep Co-Clustering (DeepCC) \cite{dongkuanxu2019DeepCoClustering}, which integrates deep autoencoders with Gaussian Mixture Models, also faces efficiency challenges with diverse data types and large datasets.

\subsection{Parallelizing Co-clustering}

Parallel co-clustering methods have emerged as a vital solution to the challenges of processing big data. The CoClusterD framework by Cheng \textit{et al.} \cite{cheng2015CoClusterDDistributedFramework} utilizes an Alternating Minimization Co-clustering (AMCC) algorithm with sequential updates in a distributed environment. However, this method faces challenges with guaranteed convergence, leading to potential inefficiencies.

While matrix factorization techniques have shown promise for co-clustering large datasets, scaling to massive high-dimensional data remains an open challenge. Chen \textit{et al.}\cite{chen2023ParallelNonNegativeMatrix} proposed a parallel non-negative matrix tri-factorization method that distributes computation across multiple nodes to accelerate factorizations. However, even these advanced methods encounter difficulties with extremely large datasets.

Our proposed method adopts a divide-and-conquer strategy, partitioning the input matrix into smaller submatrices, which are then co-clustered in parallel. This technique reduces the complexity imposed by high dimensionality and combines the results to form the final co-clusters. This novel approach addresses the computational challenges and introduces a scalable solution for big data.

\section{Mathematical Formulation and Problem Statement}\label{sec:formula}
\subsection{Mathematical Formulation of Co-clustering}
Co-clustering groups rows and columns of a data matrix $\mathbf{A} \in \mathbb{R}^{M \times N}$, where $M$ is the number of features and $N$ is the number of samples. Each element $a_{ij}$ represents the $i$-th feature of the $j$-th sample. The goal is to partition $\mathbf{A}$ into $k$ row clusters and $d$ column clusters, creating $k \times d$ homogeneous submatrices $\mathbf{A}_{I, J}$.

When rows and columns are optimally reordered, $\mathbf{A}$ can be visualized as a block-diagonal matrix, where each block represents a co-cluster with higher similarity within than between blocks. We define the row and column label sets as \( u \in \{1,\dots,k\}^M \) and \( v \in \{1,\dots,d\}^N \), respectively. Indicator matrices \( R \in \mathbb{R}^{M \times k} \) and \( C \in \mathbb{R}^{N \times d} \) are used to assign rows and columns to clusters, with the constraints \( \sum_k R_{i,k} = 1 \) and \( \sum_d C_{j,d} = 1 \), ensuring each row and column is assigned to exactly one cluster.

\subsection{Notation Clarification}
Below is a table summarizing the notations used in the mathematical formulations of our scalable co-clustering method.

\begin{table*}[h]
    \centering
    \begin{tabular}{c|p{10cm}}
        \hline
        \textbf{Symbol}        & \textbf{Description}                                                                                                           \\
        \hline
        $\mathbf{A}$           & Data matrix of dimensions $M \times N$, where $M$ is the number of rows (features) and $N$ is the number of columns (samples). \\
        $a_{ij}$               & Element at the $i$-th row and $j$-th column of matrix $\mathbf{A}$.                                                            \\
        $I, J$                 & Indices of rows and columns selected for co-clustering.                                                                        \\
        $\mathbf{A}_{I, J}$    & Submatrix containing the rows indexed by $I$ and columns by $J$.                                                               \\
        $R, C$                 & Indicator matrices for row and column cluster assignments.                                                                     \\
        $\phi_i, \psi_j$       & Block sizes in rows and columns, respectively.                                                                                 \\
        $s_i^{(k)}, t_j^{(k)}$ & Minimum row and column sizes of co-cluster $C_k$ in block $B_{(i,j)}$.                                                         \\
        $P(\omega_k)$          & Probability of failure to identify co-cluster $C_k$.                                                                           \\
        $T_p$                  & Number of sampling times or iterations in the probabilistic model.                                                             \\
        \hline
    \end{tabular}
    \caption{Notations used in the mathematical formulation of co-clustering}
    \label{tab:notations}
\end{table*}

\subsection{Problem Statement}
This paper aims to develop a method that efficiently and accurately identifies co-clusters $\mathbf{A}_{I, J}$ within a matrix $\mathbf{A}$ representing large datasets. These co-clusters should exhibit specific structural patterns such as uniformity across elements, consistency along rows or columns, or patterns demonstrating additive or multiplicative coherence. Properly identifying and categorizing these patterns is crucial for understanding the complex data structures inherent in large datasets. This method is intended to improve the detection capabilities of co-clustering, enhancing both the efficiency and precision necessary for handling large-scale data challenges.

\section{The Scalable Co-clustering Method}
\label{sec:method}
\subsection{Overview}

\begin{figure*}[htbp]
    \centering
    \includegraphics[width=0.8\linewidth]{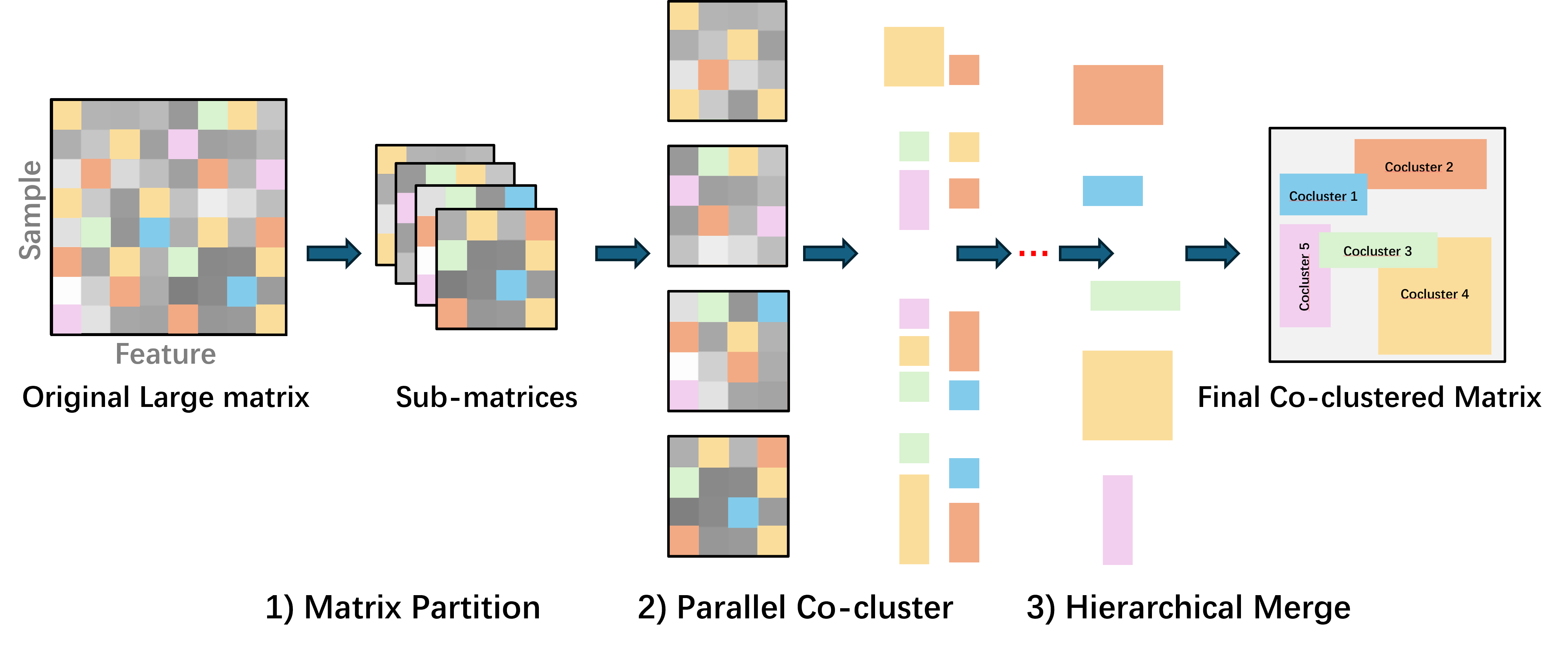}
    \caption{Workflow of our proposed Large-scale Adaptive Matrix Co-clustering for large matrices.}
    \label{fig:workflow}
\end{figure*}
This paper presents a novel and scalable co-cluster method specifically designed for large matrices, as shown in \Cref{fig:workflow}. This method applies a probabilistic model-based optimal partitioning algorithm, which not only predicts the ideal number and sequence of partitions for maximizing computational efficiency but also ensures the effectiveness of the co-clustering process.

Our method involves partitioning large matrices into smaller, manageable submatrices. This strategic partitioning is meticulously guided by our algorithm to facilitate parallel processing. By transforming the computationally intensive task of co-clustering a large matrix into smaller, parallel tasks, our approach significantly reduces computational overhead and enhances scalability.

Following the partitioning, each submatrix undergoes a co-clustering process. This is implemented via the application of Singular Value Decomposition (SVD) and $k$-means clustering on the resulting singular vectors. This pivotal step ensures the adaptability of our method, allowing our algorithm to tailor its approach to the unique characteristics of each submatrix, thus optimizing clustering results.

Furthermore, our method integrates a novel hierarchical merging strategy that combines the co-clustering results from all submatrices. This integration provides more fine-grained insight into each submatrix and enhances the overall accuracy and reliability of the co-clustering results. Our method, validated and optimized through a comprehensive process, showed efficiency in handling large-scale datasets that were never reached before.

\subsection{Large Matrix Partitioning}
The primary challenge in co-clustering large matrices is the risk of losing co-clusters when the matrix is partitioned into smaller submatrices. To address this, we introduce an optimal partitioning algorithm underpinned by a probabilistic model. This model is meticulously designed to navigate the complexities of partitioning, ensuring that the integrity of co-clusters is maintained even as the matrix is divided. The objective of this algorithm is twofold: to determine the optimal partitioning strategy that minimizes the risk of fragmenting significant co-clusters and to define the appropriate number of repartitioning iterations needed to achieve a desired success rate of co-cluster identification.

\subsubsection{Partitioning and Repartitioning Strategy based on the Probabilistic Model}
Our probabilistic model serves as the cornerstone of the partitioning algorithm. It evaluates potential partitioning schemes based on their ability to preserve meaningful co-cluster structures within smaller submatrices. The model operates under the premise that each atom-co-cluster (the smallest identifiable co-cluster within a submatrix) can be identified with a probability $p$. This probabilistic model allows us to estimate the likelihood of successfully identifying all relevant co-clusters across the partitioned submatrices.

In the scenario where the matrix $A$ is partitioned into $m \times n$ blocks, each block has size $\phi_i \times \psi_j$, that is, $M=\sum_{i=1}^m \phi_i$ and $N=\sum_{j=1}^n \psi_j$, the joint probability of $M_{(i,j)}^{(k)}$ and $N_{(i,j)}^{(k)}$ are
\begin{equation}
    \label{eq:joint_probability}
    \begin{split}
        P(M_{(i,j)}^{(k)} & < T_m, N_{(i,j)}^{(k)} < T_n)                                                                           \\
                          & = \sum_{\alpha=1}^{T_m-1} \sum_{\beta=1}^{T_n-1} P(M_{(i,j)}^{(k)} = \alpha) P(N_{(i,j)}^{(k)} = \beta) \\
                          & \le \exp[-2 (s_i^{(k)})^2 \phi_i + -2 (t_j^{(k)})^2 \psi_j]
    \end{split}
\end{equation}
where $s_i^{(k)}$ and $t_j^{(k)}$ are the minimum row and column sizes of co-cluster $C_k$ in block $B_{(i,j)}$, the size of the co-cluster $C_k$ is $M^{(k)} \times N^{(k)}$, and $M^{(k)}$ and $N^{(k)}$ are the row and column sizes of co-cluster $C_k$, respectively.


Thus, the probability of identifying all co-clusters is given by

\begin{equation}
    \begin{split}
        P(\omega_k) & \le \exp \left\{ -2 [\phi m (s^{(k)})^2 + \psi n (t^{(k)})^2] \right\},
    \end{split}
\end{equation}
and
\begin{equation}
    \begin{split}
        P & = 1 - P(\omega_k)^{T_p}                                                                                                       \\
          & \ge 1 - \exp \left\{ -2 T_p [\phi m (s^{(k)})^2 + \psi n (t^{(k)})^2] \right\} \label{eq:prob_of_identifying_all_co_clusters}
    \end{split}
\end{equation}
where $P(\omega_k)$ is the probability of the failure of identifying co-cluster $C_k$, $T_p$ is the number of sampling times, $\phi$ and $\psi$ are the row and column block sizes, and $s^{(k)}$ and $t^{(k)}$ are the minimum row and column sizes of co-cluster $C_k$.

\Cref{eq:prob_of_identifying_all_co_clusters} is central to our algorithm for partitioning large matrices for co-clustering, providing a probabilistic model that informs and optimizes our partitioning strategy to preserve co-cluster integrity. It mathematically quantifies the likelihood of identifying all relevant co-clusters within partitioned blocks, guiding us to mitigate risks associated with partitioning that might fragment vital co-cluster relationships.

Based on \eqref{eq:prob_of_identifying_all_co_clusters}, we can establish a constraint between the repartitioning time $Tr$ and the partition solution $Part$, ensuring that the partitioning strategy adheres to a predetermined tolerance success rate, thereby minimizing the risk of co-cluster fragmentation. The constraints are discussed in appendix due to space limitations.

\subsubsection{Optimization and Computational Efficiency}
Optimizing the partitioning process for computational efficiency is paramount in both academic and industrial applications, where running time often serves as the primary bottleneck. Thanks to the flexible framework established by our probabilistic model and the constraints derived in \Cref{thm:probability_co_cluster_detection}, our optimization strategy can be tailored to address the most critical needs of a given context. In this case, we focus on minimizing the running time without compromising the integrity and success rate of co-cluster identification.

Our approach to optimization leverages the probabilistic model to assess various partitioning configurations, balancing the trade-off between computational resource allocation and the need to adhere to theoretical success thresholds. By systematically evaluating the impact of different partitioning schemes on running time, we can identify strategies that not only meet our co-clustering success criteria but also optimize the use of computational resources.

To ensure that our optimization does not sacrifice the quality of co-cluster identification for the sake of efficiency, we introduce a set of conditions under which optimization can be achieved without compromising the success rate of co-cluster discovery. These conditions provide a mathematical basis for optimizing the partitioning algorithm in a manner that maintains a balance between computational efficiency and the fidelity of co-cluster identification.

Under the constraint of maintaining a predetermined success rate \(P\) for co-cluster identification, the optimization of the partitioning algorithm with respect to running time must satisfy the following condition:
\begin{equation}
    \label{eq:optimization_condition}
    \begin{split}
        T_p = \text{argmin}_{T_p} \{
        1 & - \exp \{ -2 T_p [\phi m (s^{(k)})^2         \\
          & + n (t^{(k)})^2] \} \ge P_{\text{thresh}} \}
    \end{split}
\end{equation}

This condition delineates the parameters within which the partitioning strategy can be optimized for speed without detracting from the algorithm's ability to accurately identify co-clusters. By adhering to these conditions, we ensure that our optimization efforts align with the overarching goal of preserving the integrity and effectiveness of co-cluster discovery. This balance is crucial for developing a partitioning algorithm that is not only fast and efficient but also robust and reliable across various data sets and co-clustering challenges.

\subsection{Co-clustering on Small Submatrices}

\subsubsection{Atom-co-clustering Algorithm}
Our framework, which encompasses both partitioning and ensembling, offers remarkable flexibility, allowing it to be compatible with a wide range of atom-co-clustering methods. For the sake of making this paper comprehensive and self-contained, we provide an introduction to the atom-co-cluster method herein. The only requirement for an atom-co-clustering method to be compatible with our framework is that it must be able to identify co-clusters under a given criterion with a probability $p$, or more relaxed conditions, has a lower bound estimate of the probability of identifying co-clusters equipped with a validation mechanism.

\subsubsection{Graph-based Spectral Co-clustering Algorithm}

Spectral co-clustering (SCC) stands as one of the most prevalent methods in the realm of co-clustering today\cite{vonluxburg2007TutorialSpectralClustering}, primarily due to its adeptness in unraveling the complexities of high-dimensional data. At its core, this method harnesses the power of spectral clustering principles, focusing on the utilization of a graph's Laplacian matrix eigenvectors for effectively partitioning data into distinct clusters. This approach is exceptionally beneficial for analyzing data that naturally forms a bipartite graph structure, including applications in text-document analysis, social network modeling, and gene expression studies.

\paragraph{Graph Construction in Co-clustering Expanded}

SCC begins with constructing a bipartite graph $G=(U,V,E)$. Here, $U$ and $V$, both as vertex sets, symbolize the sets corresponding to the rows and columns of the data matrix, respectively. The edges $E$ of this graph are assigned weights reflecting the relationships between rows and columns. Consequently, the graph's weighted adjacency matrix $W$ is defined as:

\begin{equation} W = \begin{bmatrix} 0 & A \\ A^T & 0 \end{bmatrix}, \end{equation}

where $A$ denotes the data matrix, also called adjacency matrix in the graph context.
Through this representation, the challenge of co-clustering is reformulated into a graph partitioning task, aiming to segregate the graph into distinct clusters based on the interrelations between the data matrix's rows and columns.

\paragraph{Laplacian Matrix}

The graph's Laplacian matrix $L$ is computed as $L=D-W$, with $D$ being the graph's degree matrix—a diagonal matrix whose entries equal the sum of the weights of the edges incident to each node. The Laplacian matrix plays a crucial role in identifying the graph's cluster structure. It does so by facilitating the calculation of eigenvectors associated with its smallest positive eigenvalues, which in turn, guide the partitioning of the graph into clusters.

\paragraph{Graph Partitioning and Singular Value Decomposition}
Theorem 4 in \cite{dhillon2001CoclusteringDocumentsWords} states that the eigenvector corresponding to the second smallest eigenvalue of the following eigenvalue problem gives the generalized partition vectors for the graph:

\begin{equation}
    L \mathbf{v} = \lambda D \mathbf{v}
    \label{eq:eigenvalue_problem}
\end{equation}

And according to Section 4 of \cite{dhillon2001CoclusteringDocumentsWords}, the singular value decomposition of the normalized matrix $A_n = D^{-1/2} A D^{-1/2}$
\begin{equation}
    \begin{split}
        A_n = U \Sigma V^T
    \end{split}
\end{equation}
gives the solution to \eqref{eq:eigenvalue_problem}. To be more specific, the singular vectors corresponding to the second largest singular value of $A_n$ is the eigenvector corresponding to the second smallest eigenvalue of \eqref{eq:eigenvalue_problem}.

The above discussion is under the assumption that the graph has only one connected component. In a more general setting, $\mathbf{u_2}, \mathbf{u_3}, \ldots, \mathbf{u_{l+1}}$ and $\mathbf{v_2}, \mathbf{v_3}, \ldots, \mathbf{v_{l+1}}$ reveal the $k$-modal information of the graph, where $\mathbf{u_k}$ and $\mathbf{v_k}$ are the $k$-th left and right singular vectors of $A_n$, respectively.
And for the last step,
\begin{equation} Z = \begin{bmatrix} D_1^{-1/2} \hat{U} \\ D_2^{-1/2} \hat{V} \end{bmatrix} \end{equation}

is stacked where $\hat{U} = [\mathbf{u_2}; \mathbf{u_3}; \ldots; \mathbf{u_{l+1}}]$ and $\hat{V} = [\mathbf{v_2}; \mathbf{v_3}; \ldots; \mathbf{v_{l+1}}]$. The approximation to the graph partitioning optimization problem is then solved by applying a $k$-means algorithm to the rows of $Z$. More details can be found in \cite{dhillon2001CoclusteringDocumentsWords}.

\subsection{Hierarchical Co-cluster Merging}

Hierarchical co-cluster merging is a novel approach that combines the results of co-clustering on submatrices to produce a final co-clustered result.
The merging method is designed to enhance the accuracy and robustness of the co-clustering outcome by leveraging the design of the partitioning algorithm. The hierarchical merging process iteratively combines the co-clusters from each submatrix, ensuring that the final co-clustered result is comprehensive and consistent across all submatrices. This iterative merging process is crucial for addressing issues of heterogeneity and model uncertainty, ensuring that the final co-clustering results are reliable and robust.

\subsection{Algorithmic Description}
Our proposed  Optimal Matrix Partition and Hierarchical Co-cluster Merging Method is outlined in Algorithm \ref{alg:method}. The algorithm
is an advanced algorithm designed for efficient co-clustering of large data matrices. The algorithm begins by initializing a block set based on predetermined block sizes. For each co-cluster in the given set, the algorithm calculates specific values $s^{(k)}$ and $t^{(k)}$, which are then used to determine the probability $P(\omega_k)$ of each co-cluster. If this probability falls below a predefined threshold $P_{\text{thresh}}$, the algorithm partitions the data matrix $A$ into blocks $B$ and performs co-clustering on these blocks. This step is crucial for managing large datasets by breaking them down into smaller, more manageable units. After co-clustering, the results from each block are aggregated to form the final co-clustered result $\mathcal{C}$. The algorithm's design allows for a flexible and efficient approach to co-clustering, particularly suited to datasets with high dimensionality and complexity.

\begin{algorithm}[!t]
    \caption{Optimal Matrix Partition and Hierarchical Co-cluster Merging Method}\label{alg:method}
    \KwIn{Data matrix $A \in \mathbb{R}^{M \times N}$, Co-cluster set $C = \{C_k\}_{k=1}^K$, Block sizes $\{\phi_i\}_{i=1}^m$, $\{\psi_j\}_{j=1}^n$, Thresholds $T_m$, $T_n$, Sampling times $T_p$, Probability threshold $P_\text{thresh}$;}
    \KwOut{Co-clustered result $\mathcal{C}$;}
    Initialize block set $B = \{B_{(i,j)}\}_{i=1}^m,_{j=1}^n$ based on $\phi_i$ and $\psi_j$\;
    Calculate $s^{(k)}$ and $t^{(k)}$ for each co-cluster $C_k$\;
    \For{$k=1$ \KwTo $K$}{
        Calculate $P(\omega_k)$ for co-cluster $C_k$\;
        \If{$P(\omega_k) < P_\text{thresh}$}{
            Partition matrix $A$ into blocks $B$ and perform co-clustering\;
            Aggregate co-clustered results from each block\;
        }
    }
    \Return Aggregated co-clustered result $\mathcal{C}$\;
\end{algorithm}


\begin{table*}[h!]
    \centering
    \caption{Comparison of Running Times (in seconds) for Various Co-clustering Methods on Selected Datasets.}
    \label{tab:running-time}
    \begin{tabular}{@{} l cccc @{}}
        \toprule
        Dataset     & SCC \cite{dhillon2001CoclusteringDocumentsWords} & PNMTF \cite{chen2023ParallelNonNegativeMatrix} & \textbf{LAMC-SCC} & \textbf{LAMC-PNMTF} \\
        \midrule
        Amazon 1000 & 64545.2                                          & 303.7                                          & 112.5             & 242.8               \\
        CLASSIC4    & *                                                & 17,810                                         & 22,894            & 3,028               \\
        RCV1-Large  & *                                                & 277,092                                        & *                 & 208,048             \\
        \bottomrule
    \end{tabular}
    \begin{tablenotes}
        \small
        \item Notes: * indicates that the method cannot process the dataset because the dataset size exceeds the processing limit.
    \end{tablenotes}
\end{table*}

\begin{table*}[htbp]
    \centering
    \caption{NMIs and ARIs Scores for Various Co-clustering Methods on Selected Datasets.}
    \label{tab:evaluation-metrics}
    \begin{tabular}{@{} l c cccc @{}}
        \toprule
        \multirow{2}{*}{Dataset}     & \multirow{2}{*}{Metric} & \multicolumn{4}{c}{Compared Methods}                                                                                                        \\
        \cmidrule{3-6}
                                     &                         & SCC \cite{dhillon2001CoclusteringDocumentsWords} & PNMTF \cite{chen2023ParallelNonNegativeMatrix} & \textbf{LAMC-SCC} & \textbf{LAMC-PNMTF} \\
        \midrule
        \multirow{2}{*}{Amazon 1000} & NMI                     & 0.9223                                           & 0.6894                                         & 0.8650            & 0.6609              \\
                                     & ARI                     & 0.7713                                           & 0.6188                                         & 0.7763            & 0.6057              \\
        \multirow{2}{*}{CLASSIC4}    & NMI                     & *                                                & 0.5944                                         & 0.7676            & 0.6073              \\
                                     & ARI                     & *                                                & 0.4523                                         & 0.5845            & 0.4469              \\
        \multirow{2}{*}{RCV1-Large}  & NMI                     & *                                                & 0.6519                                         & 0.8349            & 0.6348              \\
                                     & ARI                     & *                                                & 0.5383                                         & 0.7576            & 0.5298              \\
        \bottomrule
    \end{tabular}
    \begin{tablenotes}
        \small
        \item Notes: * indicates that the method cannot process the dataset because the dataset size exceeds the processing limit.
    \end{tablenotes}
\end{table*}

\section{Experimental Evaluation}
\label{sec:experiment}
\subsection{Experiment Setup}

\textbf{Datasets.}
The experiments were conducted using three distinct datasets to demonstrate the versatility and robustness of our method:

\begin{itemize}
    \item Amazon 1000 \cite{ni2019justifying}: Comprising 1000 Amazon reviews; each represented as a 1000-dimensional vector, this dataset is designed to mimic customer behavior analysis.
    \item CLASSIC4 \cite{reddy2021weclustering}: Containing 18000 documents from 20 newsgroups; each document is represented as a 1000-dimensional vector, this dataset is suitable for text analysis and topic discovery.
    \item RCV1-Large \cite{lewis2004rcv1}: A larger dataset used to test the scalability of our method, it includes a vast array of document vectors for high-dimensional data analysis.
\end{itemize}

\textbf{Implementation details.}
All experiments were performed on a computing cluster with the following specifications: Intel Xeon E5-2670 v3 @ 2.30GHz processors, 128GB RAM, and Ubuntu 20.04 LTS operating system. The algorithms were implemented in Rust and compiled with the latest stable version of the Rust compiler.

\textbf{Compared Methods.}
The experiments followed the procedure outlined in Algorithm \ref{alg:method}. The proposed method was compared with the following state-of-the-art co-clustering methods:

\begin{itemize}
    \item Spectral Co-Clustering (SCC) \cite{dhillon2001CoclusteringDocumentsWords}
    \item Parallel Non-negative Matrix Tri-Factorization (PNMTF)\cite{chen2023ParallelNonNegativeMatrix}
    \item Deep Co-Clustering (DeepCC) \cite{dongkuanxu2019DeepCoClustering}
\end{itemize}

Notably, 1) PNMTF is one of the most efficient co-clustering algorithms in the state-of-art. 2) All our experiments show that DeepCC cannot process all selected datasets due to the dataset size exceeds DeepCC processing limit.

\textbf{Our Methods.} Our proposed scalable co-cluster method is applied along with the SCC and PNMTF to demonstrate the enhanced performance and capability of handling large datasets:
\begin{itemize}
    \item Matrix Partitioned and Hierarchical Co-Cluster Merging with Spectral Co-Clustering (LAMC-SCC)
    \item Matrix Partitioned and Hierarchical Co-Cluster Merging with Parallel Non-negative Matrix Tri-Factorization (LAMC-PNMTF)
\end{itemize}

\textbf{Evaluation Metrics.}
The effectiveness of the co-clustering was measured using two widely accepted metrics:

\begin{itemize}
    \item Normalized Mutual Information (NMI): Quantifies the mutual information between the co-clusters obtained and the ground truth, normalized to [0, 1] range, where 1 indicates perfect overlap.
    \item Adjusted Rand Index (ARI): Adjusts the Rand Index for chance, providing a measure of the agreement between two clusters, with values ranging from $-1$ (complete disagreement) to 1 (perfect agreement).
\end{itemize}

\subsection{Results}
The experimental results are presented in Tables \ref{tab:running-time} and \ref{tab:evaluation-metrics}, comparing our methods, LAMC-SCC and LAMC-PNMTF, with traditional methods SCC and PNMTF.

\subsubsection{Handling Large-scale Datasets} The results highlight the limitations of traditional methods like SCC and DeepCC in processing large datasets, as shown by their inability to handle certain datasets (denoted by "*"). This underscores the scalability challenges in existing co-clustering methods.

\subsubsection{Improved Performance} Our methods successfully processed all datasets and significantly outperformed traditional methods in efficiency. For example, the running time for the Amazon 1000 dataset was reduced from 64545.2 seconds (SCC) to 112.5 seconds (LAMC-SCC), demonstrating a substantial increase in speed.

\subsubsection{Quantitative Metrics} As shown in Table \ref{tab:evaluation-metrics}, our methods also improved accuracy and robustness. For instance, in the CLASSIC4 dataset, LAMC-SCC achieved an NMI of 0.7676 and an ARI of 0.5845, outperforming PNMTF.

These experiments validate our proposed scalable co-clustering method as more efficient and capable of handling diverse and large-scale datasets without sacrificing the quality of co-cluster identification. The adaptability of our method to different data characteristics and its capacity for parallel processing demonstrate its potential as a robust tool for applications in domains requiring the analysis of large data matrices, such as text and biomedical data analyses and financial pattern recognition.


\section{Conclusion}
\label{sec:conclude}
This paper introduces a novel, scalable co-clustering method for large matrices, addressing the computational challenges of high-dimensional data analysis. Our method first partitions large matrices into smaller, parallel-processed submatrices, significantly reducing processing time. Next, a hierarchical co-cluster merging algorithm integrates the submatrix results, ensuring accurate and consistent final co-clustering. Extensive evaluations demonstrate that our method outperforms existing solutions in handling large-scale datasets, proving its effectiveness, efficiency, and scalability.




\newpage
\bibliographystyle{IEEEtran}
\bibliography{updated}

\newpage

\appendix

\begin{theorem}
    \label{thm:probability_co_cluster_detection}
    If the matrix $\mathbf{A}$ is partitioned into $m \times n$ blocks, each with sizes $\phi_i \times \psi_j$, and the probability of failing to detect co-cluster $\mathbf{C}_k$ in any block is $P(\omega_k)$, then
    \begin{equation}
        \begin{split}
            P(\omega_k) & \le \exp \left\{ -2 [\phi m (s^{(k)})^2 + \psi n (t^{(k)})^2] \right\}
        \end{split}
    \end{equation}
    Given $T_p$ times of random sampling, the probability of detecting the co-cluster $C_k$ is
    \begin{equation}
        \begin{split}
            P & = 1 - P(\omega_k)^{T_p}                                                        \\
              & \ge 1 - \exp \left\{ -2 T_p [\phi m (s^{(k)})^2 + \psi n (t^{(k)})^2] \right\}
        \end{split}
    \end{equation}

\end{theorem}

\begin{proof}
    Consider co-cluster $C_k$,
    \begin{equation}
        \begin{split}
            P(M_{(i,j)}^{(k)} = \alpha) & = \frac{\binom{M^{(k)}}{\alpha} \binom{M-M^{(k)}}{\phi_i-\alpha}}{\binom{M}{\phi_i}} \\
            P(N_{(i,j)}^{(k)} = \beta)  & = \frac{\binom{N^{(k)}}{\beta} \binom{N-N^{(k)}}{\psi_j-\beta}}{\binom{N}{\psi_j}}
        \end{split}
    \end{equation}
    The tail probability of $M_{(i,j)}^{(k)}$ and $N_{(i,j)}^{(k)}$ are
    \begin{equation}
        \begin{split}
            P(M_{(i,j)}^{(k)} < T_m) & = \sum_{\alpha=1}^{T_m-1} P(M_{(i,j)}^{(k)} = \alpha) \\
                                     & \le \exp(-2 (s_i^{(k)})^2 \phi_i)
        \end{split}
    \end{equation}
    where $s_i^{(k)} = \cfrac{M^{(k)}}{M}-\cfrac{T_m-1}{\phi_i}$, and
    \begin{equation}
        \begin{split}
            P(N_{(i,j)}^{(k)} < T_n) & = \sum_{\beta=1}^{T_n-1} P(N_{(i,j)}^{(k)} = \beta) \\
                                     & \le \exp (-2 (t_j^{(k)})^2 \psi_j)
        \end{split}
    \end{equation}
    where $t_j^{(k)} = \cfrac{N^{(k)}}{N}-\cfrac{T_n-1}{\psi_j}$.

    The joint probability of $M_{(i,j)}^{(k)}$ and $N_{(i,j)}^{(k)}$ are
    \begin{equation}
        \begin{split}
            P & (M_{(i,j)}^{(k)} < T_m, N_{(i,j)}^{(k)} < T_n)              \\ & = \sum_{\alpha=1}^{T_m-1} \sum_{\beta=1}^{T_n-1} P(M_{(i,j)}^{(k)} = \alpha) P(N_{(i,j)}^{(k)} = \beta) \\
              & \le \exp[-2 (s_i^{(k)})^2 \phi_i + -2 (t_j^{(k)})^2 \psi_j]
        \end{split}
    \end{equation}
    If $\phi_i = p$ and $\psi_j = q$ for all $i$ and $j$, then

    Suppose event $\omega_k$ is that co-cluster $C_k$ can't be find in any block $B_{(i,j)}$, then
    \begin{equation}
        \begin{split}
            P(\omega_k) & = \prod_{i=1}^m \prod_{j=1}^n P(M_{(i,j)}^{(k)} < T_m, N_{(i,j)}^{(k)} < T_n)                          \\
                        & \le \prod_{i=1}^m \prod_{j=1}^n \exp\{-2 \left[ (s_i^{(k)})^2 \phi_i + (t_j^{(k)})^2 \psi_j \right] \} \\
                        & = \exp\{-2 \sum_{i=1}^m \sum_{j=1}^n \left[ (s_i^{(k)})^2 \phi_i + (t_j^{(k)})^2 \psi_j \right] \}     \\
        \end{split}
    \end{equation}

    If $\phi_i = \phi$ and $\psi_j = \psi$ for all $i$ and $j$, then
    \begin{equation}
        \begin{split}
            s_i^{(k)} & = s^{(k)} = \frac{M^{(k)}}{M}-\frac{T_m-1}{\phi} \\
            t_j^{(k)} & = t^{(k)} = \frac{N^{(k)}}{N}-\frac{T_n-1}{\psi}
        \end{split}
    \end{equation}

    \begin{equation}
        \begin{split}
            P(\omega_k) & \le \exp \left\{ -2 [\phi m (s^{(k)})^2 + \psi n (t^{(k)})^2] \right\} \\
        \end{split}
    \end{equation}

    And if we do $T_p$ times of random sampling, the Probability of detecting the co-cluster is
    \begin{equation}
        \begin{split}
            P & = 1 - P(\omega_k)^{T_p}                                                        \\
              & \ge 1 - \exp \left\{ -2 T_p [\phi m (s^{(k)})^2 + \psi n (t^{(k)})^2] \right\} \\
        \end{split}
    \end{equation}
    according to which, we can set $m, n, \phi, \psi, T_m, T_n$ and $T_p$ to ensure the probability of detecting the co-cluster is larger than a given threshold.
\end{proof}

\end{document}